\documentclass[letterpaper, 10 pt, conference]{ieeeconf}  
\usepackage{cite}
\usepackage{amsmath,amssymb,amsfonts}
\usepackage{algorithmic}
\usepackage{graphicx}
\usepackage{algorithm,algorithmic}
\usepackage{hyperref}
\usepackage{subfig}
\hypersetup{hidelinks=true}
\usepackage{textcomp}
\usepackage{xcolor}
\def\BibTeX{{\rm B\kern-.05em{\sc i\kern-.025em b}\kern-.08em
    T\kern-.1667em\lower.7ex\hbox{E}\kern-.125emX}}

\newtheorem{theorem}{\textbf{Theorem}}
\newtheorem{lemma}{\textbf{Lemma}}
\newtheorem{remark}{\textbf{Remark}}
\newtheorem{assume}{\textbf{Assumption}}
\newtheorem{define}{\textbf{Definition}}
\newtheorem{corollary}{\textbf{Corollary}}
\newtheorem{problem}{\textbf{Problem}}

\usepackage{tikz}
\usepackage{tikz,tkz-euclide}

\usepackage{epstopdf} 

\IEEEoverridecommandlockouts                              

\title{\LARGE \bf
Angle-Based Formation Tracking of Underactuated Planar Agents
}

\author{Nhat-Minh Le-Phan, Yali Fan, Thien-Minh Nguyen,  \emph{Member, IEEE}, Lihua Xie*, \emph{Fellow, IEEE}
\thanks{Nhat-Minh Le-Phan and Yali Fan are with the School of Electrical and Electronic Engineering, Nanyang Technological University, Singapore 639798 (e-mails: phannhat001@e.ntu.edu.sg; fany0029@e.ntu.edu.sg).}
\thanks{Thien-Minh Nguyen is with the School of Mechanical and Mining Engineering at The University of Queensland, Brisbane, QLD 4072, Australia (e-mail: thienminh.nguyen@uq.edu.au).}
\thanks{Lihua Xie is with NTU--VinUni Joint Research Laboratory for Embodied AI and Robotics, School of Electrical and Electronic Engineering, Nanyang Technological University, Singapore 639798 and VinUniversity, Hanoi 100000, Vietnam (e-mail: elhxie@ntu.edu.sg).}
\thanks{*Corresponding author.
}}

\begin{document}

\maketitle
\thispagestyle{empty}
\pagestyle{empty}

\begin{abstract}
This paper addresses the angle-based formation tracking problem for a class of heterogeneous planar underactuated agents subject to disturbances. A representative example is a group of underactuated surface vessels (USVs) operating in the surge–sway–yaw plane, in which each vessel has three degrees of freedom but only two independent control inputs, which are surge force and yaw moment. The desired formation is characterized by a set of longitudinal offset points, referred to as \emph{hand points}, together with prescribed angular constraints among triplets of these points. The formation tracking problem is studied on a leader–follower interaction graph, assuming the leader moves at constant velocity. Under the assumption that relative velocity measurements are available, the first control law achieves asymptotic formation tracking with internal stability guarantees. Moreover, a second algorithm that does not rely on relative velocity information is introduced, which successfully drives the followers’ hand points to their desired configuration. Numerical simulations involving USVs are presented to validate the effectiveness of the proposed approaches.
\end{abstract}

\section{INTRODUCTION}
The formation control problem for multi-agent systems has been extensively studied in recent years due to its broad applicability in coordinated autonomous systems \cite{Ahn2020book}. In particular, formation control methods play a key role in enabling cooperative behaviors in unmanned aerial vehicle swarms, satellite formations, and multi-robot platforms, among many other engineering applications. Based on the inter-agent constraints used to define the desired formation, formation control strategies can be classified into displacement-based \cite{Fax2004TAC}, distance-based \cite{krick2008CDC}, ratio-of-distances-based \cite{Cao2020TAC}, bearing-based \cite{Zhao2015tac,THM2021auto}, and angle-based approaches \cite{JING2019automatica,Chen2021TAC}.

Compared with other formation control approaches, angle-based formation control has attracted increasing attention due to the fact that local angle measurements are invariant under translation, rotation, and scaling maneuvers. The pioneering work on angle-based formation control was presented in \cite{eren2003cdc}, while the theory of angle rigidity was later systematically developed in \cite{JING2019automatica}. However, these results only guarantee local stability. To address global stability, \cite{CHEN2022auto} introduced the notion of angularity together with a corresponding triangular angle rigidity theory. Subsequently, a number of works have investigated localization and formation control problems under triangular angle constraints \cite{CHEN2022auto,chenTAC2023,chen2022ieeecaa,PENG2025auto}. For example, a distributed control law for single-integrator agent models was proposed in \cite{chenTAC2023}. For double-integrator agents in a leader–follower formation, \cite{chen2022ieeecaa} proposed a PD-type control law that requires relative velocity measurements in the case where leaders move with a common constant velocity. In addition, \cite{chen2022ieeecaa} addressed maneuvering leaders with translational, rotational, and scaling motions using a control law that requires additional communication among agents. Angle-based formation control in two-dimensional space with disturbance rejection was considered in \cite{PENG2025auto}, where followers are modeled as double-integrator systems subject to matched disturbances represented by trigonometric polynomials. The internal model approach was then employed to design the control law, achieving formation tracking under constant leader velocities.

In practice, planar agents such as surface vessels operating on a horizontal plane have three degrees of freedom (surge, sway, and yaw) when heave, roll, and pitch dynamics are neglected. With only two independent control inputs, these systems are underactuated since the number of inputs is less than the configuration space dimension \cite{book_vessel}. Formation control for this class of agents was earlier studied in \cite{Lu2020IJC}. In that work, a simple star-graph structure was adopted, where the leader serves as the central node and each follower is directly connected to the leader. However, as the number of agents increases, this approach becomes impractical. Thus, formation control laws in which each agent’s position is constrained by relative measurements with its neighbors are often preferred. 

To the best of our knowledge, angle-based formation control for underactuated agents remains unexplored, which motivates the present work. In this paper, two distributed control algorithms were designed to steer a group of underactuated planar agents subject to disturbances toward a desired formation defined by angle constraints among triples of agents’ hand points and the leaders’ positions, where the leaders are assumed to move with a common constant velocity. Unlike \cite{PENG2025auto}, our proposed methods are based on the variable structure control framework. Moreover, compared with existing works \cite{chen2022ieeecaa,PENG2025auto}, this paper presents the first angle-based control law for dynamical agents with relative degree two, subject to additive disturbances, that does not require relative velocity measurements.

The remainder of this paper is organized as follows. Section~\ref{sec:background} introduces the theoretical background and formulates the problem. The main results are presented in Section~\ref{sec:main}. Simulation results for a group of surface vessels are provided in Section~\ref{sec:sim}. Finally, Section~\ref{sec:conclude} concludes the paper.
\section{Preliminaries And Problem Statement}\label{sec:background}
\textbf{Notation:} Throughout this paper, $\mathbb{R}$ denotes the set of real numbers. For a vector $\mathbf{x} \in \mathbb{R}^n$, the Euclidean norm of a vector $\mathbf{x} \in \mathbb{R}^n$ is defined as $\|\mathbf{x}\| = \sqrt{\mathbf{x}^\top \mathbf{x}}$, while $\|\mathbf{x}\|_1$ and $\|\mathbf{x}\|_{\infty}$ represent its $1$-norm and infinity norm, respectively. The operators $\mathrm{null}(\cdot)$ and $\mathrm{span}(\cdot)$ refer to the null space and the span of a matrix or a set of vectors, respectively. The symbol $\otimes$ denotes the Kronecker product, and $\mathrm{sgn}(\cdot)$ denotes the signum function. The symbol $\mathcal{R}(\cdot)$ denotes the rotation matrix in 2D. Finally, $\mathbf{1}$, $\mathbf{0}$, and $\mathbf{I}$ denote the all-ones vector, the zero vector, and the identity matrix of appropriate dimensions, respectively.
\begin{figure}[!t]
\centering
\includegraphics[width=0.7\columnwidth]{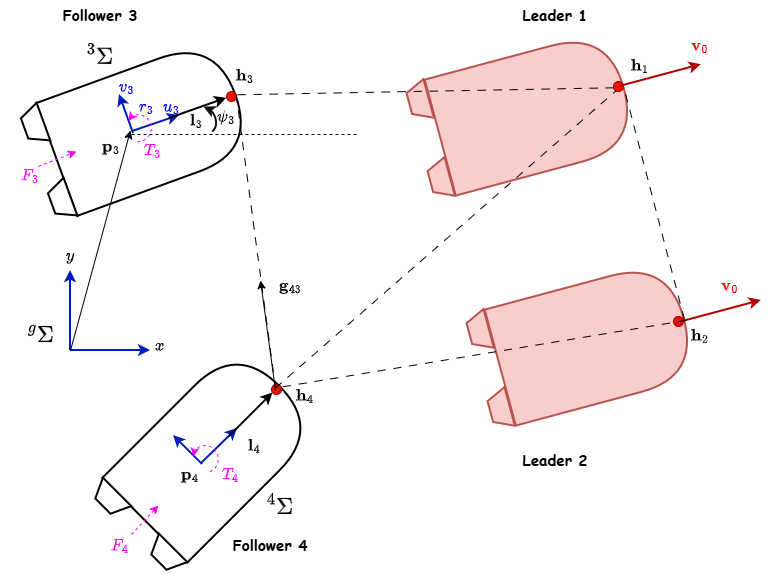}
\caption{Coordinate systems and physical quantities of a group of underactuated planar vehicles (e.g., underactuated surface vessels), including two leaders and two followers.}
\label{fig:Robot}
\end{figure}
\subsection{Mathematical Model of Underactuated Planar Agents}
Considering an underactuated planar agent, we assign to each agent $i$ a local coordinate frame ${}^i\Sigma$, whose origin is located at its center of mass or pivot point, as shown in Fig.~\ref{fig:Robot}. Let $\mathbf{p}_i=[x_i,y_i]^\top\in\mathbb{R}^2$ and $\psi_i\in\mathbb{R}$ denote the position vector and yaw angle of the agent $i$ with respect to the global coordinate ${}^g\Sigma$, respectively. The velocity of the agent $i$ expressed in its local frame and its yaw rate are denoted by ${}^i\mathbf{v}_{i}=[u_i,v_i]^\top\in\mathbb{R}^2$ and $r_i\in\mathbb{R}$, respectively. The kinematics and general nonlinear dynamics of each agent $i$, expressed in ${}^i\Sigma$, are given by
\begin{subequations}\label{eq:planar_agent}
    \begin{align}
        &\begin{aligned}
            \begin{bmatrix}\dot{\mathbf{p}}_{i}\\ \dot{\psi}_{i} \end{bmatrix} &=
            \begin{bmatrix} \cos(\psi_{i}) & -\sin(\psi_{i}) & 0 \\ \sin(\psi_{i}) & \cos(\psi_{i}) & 0 \\ 0 & 0 & 1 \end{bmatrix} 
            \begin{bmatrix} u_{i} \\ v_i \\ r_i \end{bmatrix} = 
            \begin{bmatrix} \mathcal{R}_i & \mathbf{0} \\ \mathbf{0} & 1 \end{bmatrix} \begin{bmatrix} {}^i\mathbf{v}_{i} \\ r_i \end{bmatrix},
        \end{aligned} \label{eq:planar_agent_a} \\
        &\begin{aligned}
            \begin{cases}
                \dot{u}_{i} = g_{ui}(u_i, v_i, r_{i}) + \delta_{ui}(t) + b_{ui}F_{i}, \\
                \dot{v}_{i} = g_{v i}(u_{i}, v_i, r_{i}) + \delta_{v i}(t) + b_{v i}T_{i}, \\
                \dot{r}_{i} = g_{ri}(u_i, v_i, r_{i}) + \delta_{ri}(t) + b_{ri}T_{i}.
            \end{cases}
        \end{aligned} \label{eq:planar_agent_b}
    \end{align}
\end{subequations}

Here, $F_i$ and $T_i$ represent the longitudinal control force and yaw moment of the $i$-th vessel, respectively, while $b_{ui}$, $b_{vi}$, and $b_{ri}$ denote constant actuator gain coefficients. The functions $g_{ui}$, $g_{vi}$, and $g_{ri}$ are assumed to be locally Lipschitz continuous, while $\delta_{ui}$, $\delta_{vi}$, and $\delta_{ri}$ denote bounded disturbances. In addition, $\mathcal{R}_i$ denotes $\mathcal{R}(\psi_i)$, representing the rotation matrix from the coordinate frame ${}^i\Sigma$ to ${}^g\Sigma$. The lateral and yaw dynamics, corresponding to the second and third equations in \eqref{eq:planar_agent_b}, cannot be controlled independently, as both are influenced by the single control input $T_i$. In addition, the general nonlinear functions $g_{ui}$, $g_{vi}$, and $g_{ri}$ are sufficient to describe various classes of planar underactuated systems. For example, in this paper, we are particularly interested in the model of an underactuated surface vessel \cite{book_vessel}, which is also considered in the simulation study
\begin{equation}\label{eq:vessel}
\footnotesize
\begin{aligned} 
&\underbrace{\begin{bmatrix} m_{11,i} & 0 & 0\\ 0 & m_{22,i} & m_{23,i}\\ 0 & m_{23,i} & m_{33,i}\end{bmatrix}}_{\mathcal{M}_i}\begin{bmatrix}\dot{u}_{i}\\\dot{v}_{i}\\\dot{r}_{i}\end{bmatrix}+\begin{bmatrix} -m_{22,i}v_i r_{i}-m_{23,i}r_{i}^{2}\\ m_{11,i}u_{i}r_{i}\\ (m_{22,i}-m_{11,i})u_{i}v_{i}+m_{23,i}u_{i}r_{i}\end{bmatrix}\\ &+\begin{bmatrix} d_{11,i} & 0 & 0\\ 0 & d_{22,i} & d_{23,i}\\ 0 & d_{23,i} & d_{33,i}\end{bmatrix}\begin{bmatrix} u_{i}\\v_{i}\\ r_{i}\end{bmatrix}=\begin{bmatrix} F_{i}\\ 0\\ T_{i}\end{bmatrix}+\begin{bmatrix} w_{ui}\\ w_{v_i}\\ w_{ri}\end{bmatrix},
\end{aligned}
\end{equation}
where $m(\cdot)$ represents the inertia parameters of the vessel, and $d(\cdot)$ denotes the hydrodynamic damping coefficients. The inputs $F_i$ and $T_i$ correspond to the surge force and yaw moment, respectively. The term $w(\cdot)$ represents both bounded model uncertainties and external disturbances. The above equation can be transformed into the form of (\ref{eq:planar_agent}) by multiplying both sides by $\mathcal{M}_i^{-1}$. 

Finally, for each agent $i$, a hand point $\mathbf{h}_i$ is selected along its local $x$-body axis, measured with respect to its center of mass $\mathbf{p}_i$, such that the corresponding offset vector is $\mathbf{l}_i=\mathbf{h}_i-\mathbf{p}_i$ with $\|\mathbf{l}_i\|=l_i$. Let ${}^i\mathbf{v}_{h_i}\in\mathbb{R}^2$ and ${}^i\mathbf{a}_{h_i}\in\mathbb{R}^2$ denote the velocity and acceleration of the hand point $\mathbf{h}_i$, respectively, expressed in ${}^i\Sigma$. Then, the following relationships hold \cite{quoc2023TCNS}
\begin{equation*} {}^{i}\mathbf{v}_{hi}=\mathcal{R}^{\top}_i\dot{\mathbf{h}}_{i},\ \  {}^{i}\mathbf{a}_{hi}=\mathcal{R}^{\top}_i\ddot{\mathbf{h}}_{i},\quad \forall i\in \mathcal{V}.\end{equation*}

From the relationship $\mathbf{h}_i=\mathbf{p}_i+\mathbf{l}_i$, we have
\begin{align*}
{}^i\mathbf{v}_{h i} & =\left[\begin{array}{c} u_i \\v_i+r_i l_i\end{array}\right], \\
{}^i\mathbf{a}_{h i} & =\left[\begin{array}{c}\dot{u}_i-r_i^2 l_i \\\dot{v}_i+\dot{r}_i l_i\end{array}\right] \\ & =\left[\begin{array}{c} g_{ui}+\delta_{u i}-r_i^2 l_i+b_{u i} F_i \\ g_{vi}+l_i g_{ri}+\delta_{v i}+l_i \delta_{r i}+\left(b_{v i}+b_{r i} l_i\right) T_{i}\end{array}\right].
\end{align*}

\subsection{Triangular Angle Rigidity and Problem Formulation}
The interactions among agents in the 2D plane are modeled by a communication graph $\mathcal{G} = (\mathcal{V}, \mathcal{E})$, where $\mathcal{V} = \{1, \dots, n\}$ is the vertex set and $\mathcal{E} \subseteq \mathcal{V} \times \mathcal{V}$ is the set of $m$ edges. The graph $\mathcal{G}$ is undirected if $(i, j) \in \mathcal{E} \iff (j, i) \in \mathcal{E}$, and the neighborhood of agent $i$ is denoted by $\mathcal{N}_i = \{j \in \mathcal{V} : (i, j) \in \mathcal{E}\}$. Among the \(n\) agents, the first \(n_l \geq 2\) agents are designated as leaders, while the remaining \(n_f = n - n_l\) agents are designated as followers. For any edge $(i,j)\in\mathcal{E}$, the relative displacement between the associated hand points is defined as $\mathbf{h}_{ij}=\mathbf{h}_j-\mathbf{h}_i\in\mathbb{R}^2$. A formation is defined by the pair $(\mathcal{G}, \mathbf{h})$, where 
$\mathbf{h} = [\mathbf{h}_1^{\top}, \dots, \mathbf{h}_n^{\top}]^{\top} \in \mathbb{R}^{2n}$ 
represents the collective configuration vector. This vector can be partitioned as 
$\mathbf{h} = [(\mathbf{h}^L)^{\top}, (\mathbf{h}^F)^{\top}]^{\top}$, where 
$\mathbf{h}^L=[\mathbf{h}_1^\top,\dots,\mathbf{h}_{n_L}^\top]^\top$ and 
$\mathbf{h}^F=[\mathbf{h}_{n_L+1}^\top,\dots,\mathbf{h}_{n}^\top]^\top$ represent 
the leader and follower configurations, respectively.

For three distinct agents $k, i$, and $j$, the interior angle $\alpha_{kij} \in [0, 2\pi)$ at vertex $i$ is defined as \cite{CHEN2022auto}
\begin{equation} \label{eq:angle_def}
\alpha_{kij} = \begin{cases} 
\arccos \left( \mathbf{g}_{ij}^\top \mathbf{g}_{ik} \right),  \ \ \text{if } \mathbf{g}_{ij}^\top \mathcal{R}(\pi/2)\mathbf{g}_{ik} \ge 0 \\ 
2\pi - \arccos \left( \mathbf{g}_{ij}^\top \mathbf{g}_{ik} \right),  \text{otherwise} 
\end{cases}
\end{equation}
where $\mathbf{g}_{ij} = \frac{\mathbf{h}_j - \mathbf{h}_i}{\|\mathbf{h}_j - \mathbf{h}_i\|}$ denotes the unit bearing vector from the hand point of agent $i$ to the hand point of agent $j$. Based on this definition, $\alpha_{kij}$ represents the angle formed by rotating $\mathbf{g}_{ik}$ counterclockwise to $\mathbf{g}_{ij}$.

The angle set is defined as $\mathcal{A} = \{(k,i,j) \mid k,i,j \in \mathcal{V},\; k \neq i,\; i \neq j,\; j \neq k \}$. Furthermore, an \emph{angularity} is characterized by the framework $(\mathcal{V}, \mathcal{A}, \mathbf{h})$, where each ordered triplet $(k,i,j) \in \mathcal{A} \subseteq \mathcal{V}^3$ specifies an angular constraint $\alpha_{kij}$. In particular, if the angle set $\mathcal{A}$ satisfies the cyclic symmetry property such that $(k,i,j) \in \mathcal{A} \implies (i,j,k), (j,k,i) \in \mathcal{A}$, then $\mathcal{A}$ is referred to as a \emph{triangular angle set}, and the associated framework $(\mathcal{V}, \mathcal{A}, \mathbf{h})$ is called a \emph{triangular angularity}.

We consider a triangular angularity composed of $q$ triangles among $n$ agents. For a nondegenerate triangle $\triangle ijk$ (where the three vertices are not collinear), the positions $\mathbf{h}_i, \mathbf{h}_j, \mathbf{h}_k$ satisfy an angle-induced linear constraint \cite{CHEN2022auto}

{\small
\begin{equation}
\begin{split}
&f_i^{\Delta ijk}(\alpha, \mathbf{h}) = \mathbf{A}_i^{\Delta ijk}(\alpha)\mathbf{h}_i + \mathbf{A}_j^{\Delta ijk}(\alpha)\mathbf{h}_j + \mathbf{A}_k^{\Delta ijk}(\alpha)\mathbf{h}_k \\
& = \sin \alpha_{jki}(\mathbf{h}_i - \mathbf{h}_k) - \sin \alpha_{ijk} \mathcal{R}^\top (\alpha_{kij})(\mathbf{h}_i - \mathbf{h}_j) = \mathbf{0}_2,
\end{split}
\end{equation}}
\\
where $\alpha$ denotes the set of interior angles associated with $\triangle ijk$. The coefficient matrices $\mathbf{A}_i^{\triangle ijk}, \mathbf{A}_j^{\triangle ijk}, \mathbf{A}_k^{\triangle ijk} \in \mathbb{R}^{2 \times 2}$ are determined by $\alpha$ as follows
\begin{align*}
\mathbf{A}_i^{\triangle ijk}(\alpha) &= \sin \alpha_{jki} \mathbf{I}_2 - \sin \alpha_{ijk} \mathcal{R}^\top(\alpha_{kij}), \\
\mathbf{A}_j^{\triangle ijk}(\alpha) &= \sin \alpha_{ijk} \mathcal{R}^\top(\alpha_{kij}), \\
\mathbf{A}_k^{\triangle ijk}(\alpha) &= -\sin \alpha_{jki} \mathbf{I}_2.
\end{align*}

Following an analogous derivation, all angle-induced linear equations corresponding to a triangular angularity $\mathbb{A}(\mathcal{V},\mathcal{A},\mathbf{h})$, with a generic configuration $\mathbf{h}$\footnote{A position vector $\mathbf{h}$ is called generic if all of its components are algebraically independent~\cite{Chen2021TAC}.}, can be compactly represented as $\mathbf{R}_{\mathcal{A}}(\alpha)\mathbf{h}= \mathbf{0}$,
where $\mathbf{R}_{\mathcal{A}}(\alpha)\in\mathbb{R}^{2q\times2n}$ is referred to as the \emph{triangular angle rigidity matrix} given by~\cite{CHEN2022auto}
\begin{equation}
\small
\renewcommand{\arraystretch}{1}
\bordermatrix{
    & \dots & i & \dots & j & \dots & k & \dots \cr
    1\text{st}\Delta & \dots & \dots & \dots & \dots & \dots & \dots & \dots \cr
    \dots & \dots & \dots & \dots & \dots & \dots & \dots & \dots \cr
    \Delta ijk & \mathbf{0} & \mathbf{A}_i^{\Delta ijk} & \mathbf{0} & \mathbf{A}_j^{\Delta ijk} & \mathbf{0} & \mathbf{A}_k^{\Delta ijk} & \mathbf{0} \cr
    \dots & \dots & \dots & \dots & \dots & \dots & \dots & \dots \cr
    q\text{th}\Delta & \dots & \dots & \dots & \dots & \dots & \dots & \dots \cr
}.
\end{equation}
Here, each row block of $\mathbf{R}_{\mathcal{A}}(\alpha)\in\mathbb{R}^{2q\times2n}$ is associated with a triangle in the triangular angle set $\mathcal{A}$, and each column block corresponds to a vertex in the vertex set $\mathcal{V}$. The matrix $\mathbf{R}_{\mathcal{A}}(\alpha)$ has a maximum rank of $2n-4$. Moreover, \cite{chenTAC2023} shows that $\mathrm{span}\{\mathbf{1}_n\otimes \mathbf{I}_2,\bar{\mathbf{J}}_n(\mathbf{h}-\mathbf{1}_n\otimes\bar{\mathbf{h}}),\mathbf{h}-\mathbf{1}_n\otimes\bar{\mathbf{h}}\}\subseteq \mathrm{null}(\mathbf{R}_{\mathcal{A}}(\alpha))$, where $\bar{\mathbf{J}}_n=\mathbf{I}_n\otimes \mathcal{R}(\pi/2)$ and $\bar{\mathbf{h}}=\frac{1}{n}\sum_{i=1}^n\mathbf{h}_i$ denotes the centroid of the configuration.
\begin{define} \cite{CHEN2022auto}
A framework $\mathbb{A}(\mathcal{V},\mathcal{A},\mathbf{h})$ with generic $\mathbf{h}$ is called \emph{triangularly angle rigid} if and only if $\operatorname{rank}\!\left(\mathbf{R}_{\mathcal{A}}(\alpha)\right)=2n-4$. 
\end{define}

Let $\mathbf{h}^*=[(\mathbf{h}_1^*)^\top,\dots,(\mathbf{h}_n^*)^\top]^\top$ be the desired formation configuration. The following assumptions are imposed throughout the paper.
\begin{assume}\label{assume:desired_formation}
The interactions among the agents are characterized by a leader-follower graph $\mathcal{G}$, where the edges connecting followers to leaders are directed, whereas the edges among followers are undirected. Additionally, the angularity $\mathbb{A}(\mathcal{V}, \mathcal{A}, \mathbf{h}^*)$ is triangularly angle rigid. The angle constraints corresponding to the target configuration $\mathbf{h}^*$ 
are denoted by the vector $\alpha^* = [\dots, \alpha_{ijk}^*, \dots]^\top \in \mathbb{R}^{2q}$.
\end{assume}

\begin{assume}\label{assume:leader}
The leaders are initially located at their desired positions, i.e., $\mathbf{h}_i=\mathbf{h}_i^*$ for all $i=1,\dots,n_l$. Moreover, the leaders move at a common constant reference velocity $\mathbf{v}_0\in\mathbb{R}^2$, namely, $\mathbf{v}^L=\dot{\mathbf{h}}^L=\mathbf{1}_{n_l}\otimes\mathbf{v}_0$.
\end{assume}

Next, we introduce a new matrix
\begin{equation}\label{eq:lapcacian}
\begin{aligned}
\mathbf{L}(\alpha^*) =  \mathbf{R}_{\mathcal{A}}^\top(\alpha^*) \mathbf{R}_{\mathcal{A}} (\alpha^*) 
=
\begin{bmatrix}
\mathbf{L}_{ll}(\alpha^*) & \mathbf{L}_{fl}^\top(\alpha^*) \\
\mathbf{L}_{fl}(\alpha^*) & \mathbf{L}_{ff}(\alpha^*)    
\end{bmatrix},
\end{aligned}
\end{equation}
where $\mathbf{L}_{ll}(\alpha^*) \in \mathbb{R}^{2n_l\times2n_l}, \mathbf{L}_{lf} (\alpha^*)=\mathbf{L}_{fl}^\top (\alpha^*) \in \mathbb{R}^{2n_l\times2n_f},$ and $\mathbf{L}_{ff}(\alpha^*) \in \mathbb{R}^{2n_f\times2n_f}$. 
\begin{lemma}\label{lmm:desired_formation}\cite{chenTAC2023}
Given that $\mathbb{A}(\mathcal{V}, \mathcal{A}, \mathbf{h}^*)$ is triangularly angle rigid, the matrix $\mathbf{L}_{ff}(\alpha^*)$ is positive definite. Moreover, the desired positions of the followers can be determined from the desired angles $\alpha^*$ and the leaders’ positions as $\mathbf{h}^{F*}(t) = -\mathbf{L}_{ff}^{-1}(\alpha^*) \mathbf{L}_{fl}(\alpha^*) \mathbf{h}^L(t)$.
\end{lemma}
\begin{corollary}
Under Assumptions~\ref{assume:desired_formation} and~\ref{assume:leader}, the desired velocities of the followers satisfy $$\mathbf{v}^{F*}_h=\dot{\mathbf{h}}^{F*}=-\mathbf{L}_{ff}^{-1}(\alpha^*) \mathbf{L}_{fl}(\alpha^*)\dot{\mathbf{h}}^L=\mathbf{1}_{n_f}\otimes \mathbf{v}_0.$$    
\end{corollary}

The formation tracking control problem considered in this paper is formally stated as follows.
\begin{problem}\label{prob}
Under Assumptions~\ref{assume:desired_formation} and~\ref{assume:leader}, design the tracking control inputs $F_i$ and $T_i$ for each follower agent such that $\mathbf{h}(t)\rightarrow\mathbf{h}^*(t)$ as $t\rightarrow\infty$.
\end{problem}

Finally, the following lemma is useful for the subsequent analysis.
\begin{lemma}\cite{Khalil2002nonlinear}\label{lm:vanish_input}
Consider the system $\dot{\mathbf{x}}(t)=\mathbf{M}\mathbf{x}(t)+\boldsymbol{\phi}(t)$, where $\mathbf{x}(t)\in\mathbb{R}^{d}$ is the state vector, $\mathbf{M}\in\mathbb{R}^{d\times d}$ is a Hurwitz matrix, and $\boldsymbol{\phi}(t)$ is a bounded vanishing input, i.e., $\lim_{t\to\infty}\boldsymbol{\phi}(t)=\mathbf{0}_d$. Then, $\lim_{t\to\infty}\mathbf{x}(t)=\mathbf{0}_d$. Furthermore, if there exist constants $C,\alpha>0$ such that $\|\boldsymbol{\phi}(t)\|\leq Ce^{-\alpha t}$, then $\mathbf{x}(t)$ converges to $\mathbf{0}_d$ exponentially fast.
\end{lemma}
\section{Main Results}\label{sec:main}

\subsection{Angle-Based Formation Tracking Control Law}
For each agent $i$, we define an auxiliary variable in its local coordinate frame ${}^{i}\mathbf s_{i}=[s_{1i},s_{2i}]^\top=\mathcal R^\top(\psi _{i}) {\mathbf s}_{i}$ as 
\begin{equation}
\begin{aligned}
{}^{i}\mathbf s_{i} &= \mathcal{R}^\top_i\bigg[ \dot{\mathbf {h}}_{i}+k_P\mathcal{F}_i(\alpha^*,\mathbf{h}) +k_{I}\int _{0}^{t}\mathcal{F}_i(\alpha^*,\mathbf{h}(\tau))d\tau \bigg], 
\end{aligned}
\end{equation}
where $k_P, k_I>0$ are control gains, and the expansion of $\mathcal{F}_i(\alpha^*,\mathbf{x}(t))$ has the following form
\begin{equation}
\begin{aligned}
\mathcal{F}_i&(\alpha^*,\mathbf{x}(t))\\
=&\sum_{(i, j_1, k_1) \in \bar{\mathcal{A}}} (\mathbf{A}_i^{\Delta i j_1 k_1}(\alpha^*))^\top f_i^{\Delta i j_1 k_1}(\alpha^*, \mathbf{x}(t)) \\
&+ \sum_{(j_2, i, k_2) \in \bar{\mathcal{A}}} (\mathbf{A}_i^{\Delta j_2 i k_2}(\alpha^*))^\top f_i^{\Delta j_2 i k_2}(\alpha^*, \mathbf{x}(t))\\
&+ \sum_{(j_3, k_3, i) \in \bar{\mathcal{A}}} (\mathbf{A}_i^{\Delta j_3 k_3 i}(\alpha^*))^\top f_i^{\Delta j_3 k_3 i}(\alpha^*, \mathbf{x}(t)).\\
\end{aligned}
\end{equation}

Hence, the following yields
\begin{equation}
\begin{aligned} 
{}^{i}\dot{\mathbf{s}}_{i}&=\dot{\mathcal{R}}^\top(\psi_i) {\mathbf{s}}_{i}+\mathcal{R}^\top_i {\dot{\mathbf{s}}}_{i} = r_{i}\begin{bmatrix}s_{2i}\\ -s_{1i}\end{bmatrix}+\mathcal{R}^\top_i {\dot{\mathbf{s}}}_{i} \\ &= r_{i}\begin{bmatrix}s_{2i}\\ -s_{1i}\end{bmatrix}+{}^{i}\mathbf{a}_{hi}\\
&\quad+\mathcal{R}^\top_i \big(k_P\mathcal{F}_i(\alpha^*,\mathbf{v}_h(t))+k_I\mathcal{F}_i(\alpha^*,\mathbf{h}(t))\big),
\end{aligned}
\end{equation}
where $\mathbf{v}_h=\dot{\mathbf{h}}=[(\mathbf{v}_h^F)^\top,(\mathbf{v}^L)^\top]^\top$ denotes the stacked velocity vector of the hand points. The control input for each agent $i$ is designed as follows
\begin{equation}\label{eq:control_velocity}
\begin{aligned} 
&\begin{bmatrix}b_{ui}F_i\\ (b_{vi}l_{i}+b_{ri})T_i \end{bmatrix}= -k{}^{i}\mathbf {s}_{i}-\gamma  {\text{sgn}}({}^{i}\mathbf {s}_{i})- \begin{bmatrix}g_{ui}-r_{i}^{2}l_{i}\\ g_{vi} +l_{i}g_{ri}\end{bmatrix} \\ &\quad\quad\quad\quad-\mathcal{R}^\top_i \big(k_P\mathcal{F}_i(\alpha^*,\mathbf{v}_h(t))+k_I\mathcal{F}_i(\alpha^*,\mathbf{h}(t))\big), 
\end{aligned}
\end{equation}
where $k>0$, and $\gamma > 0$ is a control gain to be designed. It can be seen that the control law requires the local velocity of the hand point, as well as the relative displacements and relative velocities between neighboring agents' hand points. Let $\boldsymbol{\delta}^F=\mathbf{h}^F-\mathbf{h}^{F*}=\mathbf{h}^F+\mathbf{L}_{ff}^{-1}(\alpha^*) \mathbf{L}_{fl}(\alpha^*) \mathbf{h}^L$ denote the tracking error of the followers. The stability of the proposed control law is provided in the following theorem.
\begin{theorem}\label{thm:stability1}
Consider the Problem \ref{prob} under Assumptions \ref{assume:desired_formation} and \ref{assume:leader}. Under the control law~(\ref{eq:control_velocity}) with $\gamma > \|[\delta_{ui},\delta_{vi}+l_i\delta_{ri}]\|_{\infty}$, $\mathbf{h}(t)\to\mathbf{h}^*(t)$ \emph{exponentially fast} as $t\to\infty$. 
\end{theorem}
\begin{proof}
{Due to the discontinuity of the control law, solutions to the closed-loop system are understood in the sense of Filippov, as the signum function $\text{sgn}(\cdot)$ is measurable and locally essentially bounded.} The stability proof of the proposed control law is divided into two parts. First, we show that the auxiliary variable $\mathbf{s}_i$ converges exponentially to the origin. Based on this result, we then prove that the hand points converge to their desired configuration. For the first part, we consider the Lyapunov function $V=\frac{1}{2}\|\mathbf{s}_i\|^2=\frac{1}{2}\|{}^i\mathbf{s}_i\|^2$. {We have $\dot V \in \dot{\tilde V}$ almost everywhere, where $\dot{\tilde V}=\nabla V^\top\mathcal{K}[\dot{\mathbf{s}}_i]$ denotes the set-valued Lie derivative of $V$.} The following yields
\begin{equation*}
\begin{aligned}
\dot{\tilde V}&=-k\|{}^i\mathbf{s}_i\|^2-\gamma\|{}^i\mathbf{s}_i\|_1+{}^i\mathbf{s}_i^\top
\begin{bmatrix}
\delta_{ui}\\
\delta_{vi}+l_i\delta_{ri}
\end{bmatrix}\\
&\leq -k\|{}^i\mathbf{s}_i\|^2-(\gamma-\|[\delta_{ui},\delta_{vi}+l_i\delta_{ri}]\|_{\infty})\|{}^i\mathbf{s}_i\|_1\\
&\leq -k\|{}^i\mathbf{s}_i\|^2.
\end{aligned}
\end{equation*}
{Here, we used the fact that ${}^i\mathbf{s}_i^\top r_{i}\begin{bmatrix}s_{2i}\\ -s_{1i}\end{bmatrix}=0$, $\mathbf{x}^\top\mathbf{y}\leq \|\mathbf{x}\|_1\|\mathbf{y}\|_{\infty}$, and $\mathbf{x}\mathcal{K}[\text{sgn}(\mathbf{x})]=\{|\mathbf{x}|\}$. Moreover, it follows from the first equality that $\dot{\tilde V}$ is singleton and $\dot{V}=\dot{\tilde V}$.} Similarly, it can be shown that the auxiliary variable ${}^i\mathbf{s}_i$ converges exponentially to the origin for all followers, i.e., $\mathbf{s}^F=[\mathbf{s}_{n_l+1},\dots,\mathbf{s}_{n}]^\top \rightarrow \mathbf{0}$ exponentially fast.
\\
For the second part, let's define the following stack vectors
\begin{equation}
\begin{aligned}
\boldsymbol{\vartheta}_i = \int _{0}^{t}\mathcal{F}_i(\alpha^*,\mathbf{h}(\tau))d\tau, \boldsymbol{\vartheta}=[\boldsymbol{\vartheta}_1^\top,\dots,\boldsymbol{\vartheta}_{n}^\top]^\top.
\end{aligned}
\end{equation}
It can be seen that $\dot{\boldsymbol{\vartheta}}=\mathbf{R}_{\mathcal{A}}^\top(\alpha^*)\mathbf{R}_{\mathcal{A}}(\alpha^*)\mathbf{h}$. The following closed loop dynamic can thus be obtained
\begin{equation*}
\begin{aligned}
\dot{\boldsymbol{\vartheta}}^F&=\mathbf{L}_{fl}(\alpha^*)\mathbf{h}^L+\mathbf{L}_{ff}(\alpha^*)\mathbf{h}^F=\mathbf{L}_{ff}(\alpha^*)\boldsymbol{\delta}^F,\\
\dot{\boldsymbol{\delta}}^F&=\dot{\mathbf{h}}^F+\mathbf{L}_{ff}^{-1}(\alpha^*) \mathbf{L}_{fl}(\alpha^*)\dot{\mathbf{h}}^L\\
&=-k_P\mathbf{L}_{ff}(\alpha^*)\boldsymbol{\delta}^F-k_I\boldsymbol{\vartheta}^F+\big(\mathbf{s}^F+\mathbf{1}_{n_f}\otimes \mathbf{v}_0\big),
\end{aligned}
\end{equation*}
where $\boldsymbol{\vartheta}^F=[\boldsymbol{\vartheta}_{n_l+1}^\top,\dots,\boldsymbol{\vartheta}_{n}^\top]^\top$. By denoting $\tilde{\boldsymbol{\vartheta}}^F=\boldsymbol{\vartheta}^F-\mathbf{1}_{n_f}\otimes \frac{\mathbf{v}_0}{k_I}$, an equivalent representation of the above equations is given as follows
\begin{equation}\label{eq:closed_loop_1}
\begin{aligned}
\begin{bmatrix}
\dot{\boldsymbol{\delta}}^F\\
\dot{\tilde{\boldsymbol{\vartheta}}}^F
\end{bmatrix}
=&
\begin{bmatrix}
-k_P\mathbf{L}_{ff}(\alpha^*) &-k_I\mathbf{I}_{2n_f}\\
\mathbf{L}_{ff}(\alpha^*) &\mathbf{0}
\end{bmatrix}
\begin{bmatrix}
\boldsymbol{\delta}^F\\
\tilde{\boldsymbol{\vartheta}}^F
\end{bmatrix}
+
\begin{bmatrix}
\mathbf{s}^F\\
\mathbf{0}
\end{bmatrix}.
\end{aligned}
\end{equation}
Since $\mathbf{L}_{ff}(\alpha^*)$ is positive definite and $k_P,k_I>0$, the state matrix in \eqref{eq:closed_loop_1} is Hurwitz. Moreover, because the input $\mathbf{s}^F$ vanishes exponentially, it follows from Lemma~\ref{lm:vanish_input} that the states $\boldsymbol{\delta}^F$ and $\tilde{\boldsymbol{\vartheta}}^F$ converge to the origin exponentially fast. The proof is completed.
\end{proof}
\begin{remark}[Internal Stability]
Theorem~\ref{thm:stability1} ensures the boundedness of $\dot{\mathbf{h}}$. Since $\mathcal{R}_i^\top\dot{\mathbf{h}}_i=[u_i,\,v_i+r_il_i]^\top$, $u_i$ is bounded. Furthermore, the boundedness of $v_i$ implies that of $r_i$, and vice versa. Thus, establishing the ultimate boundedness of the sway velocity $v_i$ is sufficient to guarantee internal stability. For example, for the surface vessel model~(\ref{eq:vessel}), by following the same procedure as \cite{quoc2023TCNS}, the sway velocity $v_i$ remains bounded if $\|\mathbf{v}_0\|_{\infty}<\frac{d_{22,i}l_i-d_{23,i}}{m_{11,i}}$ and $l_i>\max\{\frac{d_{23,i}}{d_{22,i}},\frac{m_{23,i}}{m_{22,i}}\}$. The proof is similar to that in \cite{quoc2023TCNS} and is therefore omitted.
\end{remark}

\subsection{Tracking Control Law Without Relative Velocity}
The above control law requires the relative velocity information among agents. In practice, relative velocities are often obtained by differentiating relative displacement measurements and are therefore sensitive to measurement noise. Thus, in this section, we propose another distributed tracking control law that does not rely on relative velocity measurements. The feedback linearization technique will be utilized with the following virtual control signal
$$
\boldsymbol{\nu}_i=\mathcal{R}_i
\begin{bmatrix}
g_{ui}-r_i^2 l_i+b_{u i} F_i \\ g_{vi}+l_i g_{ri}+\left(b_{v i}+b_{r i} l_i\right) T_{i}
\end{bmatrix}.
$$
Hence, by denoting $\boldsymbol{\zeta}_i(t) = \mathcal{R}_i[\delta_{ui},\delta_{vi}+l_i\delta_{ri}]^\top$, we get $\ddot{\mathbf{h}}_i=\boldsymbol{\nu}_i+\boldsymbol{\zeta}_i$ for $i=n_l+1,\dots,n$. The following matrix form can be obtained
\begin{equation}
\begin{aligned}
\dot{\mathbf{h}}&=\mathbf{v}_h=\bar{\mathbf{Z}}
\begin{bmatrix}
\mathbf{0}_{2n_l}\\
\mathbf{v}^F_h-\mathbf{1}_{n_f}\otimes\mathbf{v}_0
\end{bmatrix}+\mathbf{1}_n \otimes\mathbf{v}_0,\\
\dot{\mathbf{v}}_h&=\boldsymbol{\nu}+\boldsymbol{\zeta}=
\begin{bmatrix}
\mathbf{0}_{2n_l}\\
\boldsymbol{\nu}^F+\boldsymbol{\zeta^F}
\end{bmatrix},
\end{aligned}
\end{equation}
where $\bar{\mathbf{Z}}=\mathbf{Z}\otimes \mathbf{I}_2$, $\mathbf{Z}=\begin{bmatrix} \mathbf{0}_{n_l \times n_l} & \mathbf{0}_{n_l \times n_f} \\ \mathbf{0}_{n_f \times n_l} & \mathbf{I}_{n_f} \end{bmatrix} \otimes \mathbf{I}_2$, $\boldsymbol{\nu}=[\mathbf{0}_{2n_l}^\top,\boldsymbol{\nu}_{n_l+1}^\top,\dots,\boldsymbol{\nu}_{n}^\top]^\top$ and $\boldsymbol{\zeta}=[\mathbf{0}_{2n_l}^\top,\boldsymbol{\zeta}_{n_l+1}^\top,\dots,\boldsymbol{\zeta}_{n}^\top]^\top$. For each agent, let's define the auxiliary variable $\mathbf{s}_i=-\mathcal{F}_i(\alpha^*,\mathbf{h})$. The virtual control law is designed as follows
\begin{equation}\label{eq:virtual_law}
\begin{aligned}
\boldsymbol{\nu}_i &= \mathbf{s}_i+\dot{\boldsymbol{\eta}}_i+\gamma \text{sgn}(\boldsymbol{\eta}_i-\mathbf{v}_{hi}),\\
\dot{\boldsymbol{\Upsilon}}_i&=-k_1(\boldsymbol{\Upsilon}_i-\boldsymbol{\eta}_i),\\
\dot{\boldsymbol{\eta}}_i&=k_2\mathbf{s}_i-k_3\dot{\boldsymbol{\Upsilon}}_i=k_2\mathbf{s}_i+k_1k_3(\boldsymbol{\Upsilon}_i-\boldsymbol{\eta}_i),
\end{aligned}
\end{equation}
where $\boldsymbol{\Upsilon}_i$ and $\boldsymbol{\eta}_i$ are auxiliary variables with arbitrary initial conditions, and $k_1,k_2,k_3>0$ are positive control gains. The gain $\gamma>0$ will be designed later. Compared with (\ref{eq:control_velocity}), the proposed control law only requires relative displacement information and the local velocity. The following matrix form can thus be obtained
\begin{equation}\label{eq:control2_closed}
\begin{aligned}
\boldsymbol{\nu}&=\mathbf{s}+\dot{\boldsymbol{\eta}}+\gamma\text{sgn}\big(\bar{\mathbf{Z}}(\boldsymbol{\eta}-\mathbf{v}_h)\big),\\
\mathbf{s}&=
\begin{bmatrix}
\mathbf{0}_{2n_l}\\
\mathbf{s}^F
\end{bmatrix}
=-\bar{\mathbf{Z}}\mathbf{R}_{\mathcal{A}}^\top(\alpha^*) \mathbf{R}_{\mathcal{A}} (\alpha^*)\mathbf{h},\\
\dot{\boldsymbol{\Upsilon}}&=-k_1
\begin{bmatrix}
\mathbf{0}_{2n_l}\\
\boldsymbol{\Upsilon}^F-\boldsymbol{\eta}^F
\end{bmatrix}
=-k_1\bar{\mathbf{Z}}(\boldsymbol{\Upsilon}-\boldsymbol{\eta}),\\
\dot{\boldsymbol{\eta}}&=
\begin{bmatrix}
\mathbf{0}_{2n_l}\\
\boldsymbol{\eta}^F
\end{bmatrix}
=k_2\mathbf{s}+k_1k_3\bar{\mathbf{Z}}(\boldsymbol{\Upsilon}-\boldsymbol{\eta}).
\end{aligned}
\end{equation}
\begin{theorem}
Consider the Problem~\ref{prob} under Assumptions \ref{assume:desired_formation} and \ref{assume:leader}, and assume that $\boldsymbol{\zeta}(t)$ is uniformly continuous. Under the control law~(\ref{eq:virtual_law}) with $\gamma > \sqrt{2}\|[\delta_{ui},\delta_{vi}+l_i\delta_{ri}]\|_{\infty}$, $\mathbf{h}(t) \rightarrow \mathbf{h}^*(t)$ as $t \rightarrow \infty$.  
\end{theorem}
\begin{proof}
Considering the following Lyapunov function
\begin{equation*}
\begin{aligned}
V=&\frac{1}{2}(\boldsymbol{\delta}^F)^\top\mathbf{L}_{ff}(\alpha^*)\boldsymbol{\delta}^F+\frac{1}{2}\|\boldsymbol{\eta}^F-\mathbf{v}^F_h\|^2\\
&+\frac{1}{2k_2}\|\boldsymbol{\eta}^F-\mathbf{1}_{n_f}\otimes \mathbf{v}_0\|^2+\frac{k_3}{2k_2}\|\boldsymbol{\Upsilon}^F-\mathbf{1}_{n_f}\otimes \mathbf{v}_0\|^2\\
=&\frac{1}{2}(\mathbf{h}-\mathbf{h}^*)^\top\mathbf{L}(\alpha^*)(\mathbf{h}-\mathbf{h}^*)+\frac{1}{2}\|\bar{\mathbf{Z}}(\boldsymbol{\eta}-\mathbf{v}_h)\|^2\\
&+\frac{1}{2k_2}\|\bar{\mathbf{Z}}(\boldsymbol{\eta}-\bar{\mathbf{v}}_0)\|^2+\frac{k_3}{2k_2}\|\bar{\mathbf{Z}}(\boldsymbol{\Upsilon}-\bar{\mathbf{v}}_0)\|^2,
\end{aligned}
\end{equation*}
where $\bar{\mathbf{v}}_0=\mathbf{1}_{n}\otimes \mathbf{v}_0$. {Recalling that $\bar{\mathbf{Z}}^2=\bar{\mathbf{Z}}$, we obtain $\dot V \in^{a.e.} \dot{\tilde V}$, where}
{\small
\begin{equation*}
\begin{aligned}
 \dot{\tilde V}=&(\mathbf{h}-\mathbf{h}^*)^\top\mathbf{L}(\alpha^*)(\bar{\mathbf{Z}}(\mathbf{v}_h-\bar{\mathbf{v}}_0)+\bar{\mathbf{v}}_0)\\
 &+(\boldsymbol{\eta}-\mathbf{v}_h)^\top\bar{\mathbf{Z}}(\dot{\boldsymbol{\eta}}-\dot{\boldsymbol{\eta}}-\mathbf{s}-\gamma\mathcal{K}[\text{sgn}(\bar{\mathbf{Z}}(\boldsymbol{\eta}-\mathbf{v}_h))]-\boldsymbol{\zeta})\\
 &+\frac{1}{k_2}(\boldsymbol{\eta}-\bar{\mathbf{v}}_0)^\top\bar{\mathbf{Z}}\dot{\boldsymbol{\eta}}+\frac{k_3}{k_2}(\boldsymbol{\Upsilon}-\bar{\mathbf{v}}_0)^\top\bar{\mathbf{Z}}\dot{\boldsymbol{\Upsilon}}\\
 =&-\mathbf{s}^\top(\mathbf{v}_h-\bar{\mathbf{v}}_0)-\boldsymbol{\eta}^\top\mathbf{s}+\mathbf{v}_h^\top\mathbf{s}-\gamma\|\bar{\mathbf{Z}}(\boldsymbol{\eta}-\mathbf{v}_h)\|_1\\
 &-(\boldsymbol{\eta}-\mathbf{v}_h)^\top\boldsymbol{\zeta}+(\boldsymbol{\eta}-\bar{\mathbf{v}}_0)^\top\mathbf{s}-\frac{k_1k_3}{k_2}\|\boldsymbol{\Upsilon}-\boldsymbol{\eta}\|^2\\
 =&-\gamma\|\boldsymbol{\eta}^F-\mathbf{v}^F_h\|_1-(\boldsymbol{\eta}^F-\mathbf{v}^F_h)^\top\boldsymbol{\zeta}^F-\frac{k_1k_3}{k_2}\|\boldsymbol{\Upsilon}^F-\boldsymbol{\eta}^F\|^2\\
 \leq&-(\gamma-\|\boldsymbol{\zeta}\|_{\infty})\|\boldsymbol{\eta}^F-\mathbf{v}^F\|_1-\frac{k_1k_3}{k_2}\|\boldsymbol{\Upsilon}^F-\boldsymbol{\eta}^F\|^2\\
 \leq&-\frac{k_1k_3}{k_2}\|\boldsymbol{\Upsilon}^F-\boldsymbol{\eta}^F\|^2.
\end{aligned}
\end{equation*}}
\\
{It can be seen that $\dot{\tilde V}$ is singleton and $\dot{V}=\dot{\tilde V}$.} The last inequality holds from the fact that $\|\mathcal{R}\mathbf{x}\|_{\infty}\leq\sqrt{2}\|\mathbf{x}\|_{\infty}$ for any $\mathcal{R}\in SO(2)$ and any $\mathbf{x}\in\mathbb{R}^2$. Thus, $V$ is nonincreasing and $\lim_{t \to \infty}V(t)$ exists. Since $\mathbf{v}_0$ is constant, it follows that $\boldsymbol{\delta}^F$, $\mathbf{s}^F$, $\mathbf{v}^F_h$, $\boldsymbol{\eta}^F$, and $\boldsymbol{\Upsilon}^F$ are bounded. From (\ref{eq:control2_closed}) and the assumption that $\boldsymbol{\zeta}$ is uniformly continuous, we further conclude that $\dot{\mathbf{v}}_h^F$, $\dot{\boldsymbol{\eta}}^F$, and $\dot{\boldsymbol{\Upsilon}}^F$ are bounded. Hence, $\boldsymbol{\eta}^F-\mathbf{v}^F_h$, $\boldsymbol{\Upsilon}^F-\boldsymbol{\eta}^F$, and $\dot{V}$ are uniformly continuous. By Barbalat’s lemma, $\dot{V}(t)\to0$ as $t\to\infty$, i.e., $\boldsymbol{\Upsilon}^F-\boldsymbol{\eta}^F \to \mathbf{0}_{2n_f}$, which further implies $\dot{\boldsymbol{\Upsilon}}^F \to \mathbf{0}_{2n_f}$. Moreover, from $\ddot{\boldsymbol{\Upsilon}}^F-\ddot{\boldsymbol{\eta}}^F = -k_1(1+k_3)\big(\dot{\boldsymbol{\Upsilon}}^F-\dot{\boldsymbol{\eta}}^F\big) - k_2\dot{\mathbf{s}}^F$, and the boundedness of $\dot{\mathbf{s}}^F$, it follows that $\dot{\boldsymbol{\Upsilon}}^F-\dot{\boldsymbol{\eta}}^F$ is uniformly continuous. Applying Barbalat's lemma again yields $\dot{\boldsymbol{\Upsilon}}^F-\dot{\boldsymbol{\eta}}^F \to \mathbf{0}_{2n_f}$, which together with $\dot{\boldsymbol{\Upsilon}}^F \to \mathbf{0}_{2n_f}$ implies that $\dot{\boldsymbol{\eta}}^F \to \mathbf{0}_{2n_f}$. Consequently, $\mathbf{s}^F \to \mathbf{0}_{2n_f}$, namely, $\mathbf{h}(t)\to\mathbf{h}^*(t)$ as $t\to\infty$.
\end{proof}
\section{Simulation Results}\label{sec:sim}
\begin{figure}
\centering
\begin{tikzpicture}[
roundnode/.style={circle, draw=black, thick, minimum size=2mm,inner sep= 0.25mm},
squarednode/.style={rectangle, draw=red!60, fill=red!5, very thick, minimum size=5mm},
]
 \node[roundnode]   (v4)   at   (0,0) {$4$};
    \node[roundnode, draw=red, text=red]   (v1)   at   (1.5,0) {$1$};
    \node[roundnode]   (v3)   at   (0,1.5) {$3$};
    \node[roundnode, draw=red, text=red]   (v2)   at   (1.5,1.5) {$2$};
    \draw[-,thick] (v1)--(v2)--(v3)--(v4)--(v1);
    \draw[-,thick] (v4)--(v2);
\end{tikzpicture} 
    \caption{The desired formation shape of $\mathbb{A}(\mathcal{V},\mathcal{A},\mathbf{h}^*)$.
    }\label{fid:Desired_formation}
\end{figure}
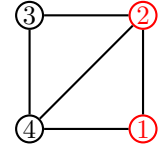
Consider a formation of four USVs modeled in (\ref{eq:vessel}) with two leaders, where the interaction graph consisting of $q=2$ triangles is shown in Fig.~\ref{fid:Desired_formation}. The triangular angle set is $\mathcal{A} = \{(1,2,4),(2,4,1),(1,4,2),(2,3,4),(3,4,2),(4,2,3)\}$. The angle constraints are $\alpha^*_{421}=\frac{\pi}{4},\alpha^*_{214}=\frac{\pi}{2},\alpha^*_{324}=\frac{\pi}{4},\alpha^*_{243}=\frac{\pi}{4}$. The following matrices can be obtained
\begin{equation*}
\small    
\begin{aligned}
\mathbf{A}_1^{\triangle 214}
&= 
\begin{bmatrix}
0.7071 &0.7071\\
-0.7071 &0.7071
\end{bmatrix},
\mathbf{A}_2^{\triangle 214}
=\begin{bmatrix}
0 &-0.7071\\
0.7071 &0
\end{bmatrix},
\\
\mathbf{A}_4^{\triangle 214}
&= 
\begin{bmatrix}
-0.7071 &0\\
0 &-0.7071
\end{bmatrix}, \mathbf{A}_2^{\triangle 243}
= 
\begin{bmatrix}
0.5 &-0.5\\
0.5 &0.5
\end{bmatrix},\\
\mathbf{A}_4^{\triangle 243}
&= 
\begin{bmatrix}
0.5 &0.5\\
-0.5 &0.5
\end{bmatrix},
\mathbf{A}_3^{\triangle 243}
= 
\begin{bmatrix}
-1 &0\\
0 &-1
\end{bmatrix}.
\end{aligned}
\end{equation*}

The leaders' hand points are initially positioned at $\mathbf{h}_1(0) = [5,5]^\top$ and $\mathbf{h}_2(0) = [5,10]^\top$. Moreover, the leaders move at a constant velocity $\mathbf{v}_0 = [0.4,0.1]^\top$. Although the proposed results are applicable to heterogeneous USVs, for simplicity, we consider a homogeneous system in which all agents share the same dynamical parameters, given as $m_{11,i}=220 \text{ kg}$, $ m_{22,i}=295 \text{ kg}$, $m_{33,i}=745\text{ kg.m}^2$, $m_{23,i}=3\text{ kgm}$, $d_{11,i}=75\text{ kg/s}$, $d_{22,i}=105\text{ kg/s}$, $d_{33,i}=50\text{ kg.m}^2\text{/s}$, $d_{23,i}=-25\text{ kgm/s}$ and $l_i=1.4 \text{ m}$ for $i=3,4$. In addition, the input disturbances are chosen as $w_{ui}=50+30\sin(2t)+30\sin(10t)$ $(N)$, $w_{vi}=0$, and $w_{ri}=100+60\sin(2t)+40\sin(10t)$ $(Nm)$.

The simulation results for the control law in \eqref{eq:control_velocity}, with $k = 3$, $k_P = 1.5$, $k_I = 0.15$, and $\gamma = 2$, are presented in Fig.~\ref{fig:PD}. Meanwhile, the simulation results for the control law in \eqref{eq:virtual_law} with $k_1 = 0.5$, $k_2 = k_3 = 5$, and $\gamma = 5$, are demonstrated in Fig.~\ref{fig:adaptive}. Overall, both control strategies successfully steer the followers' hand points along their desired trajectories, precisely achieving the prescribed angle-constrained formation.
\begin{figure}[t!]
    \centering
    \subfloat[]{\includegraphics[width=0.9\linewidth]{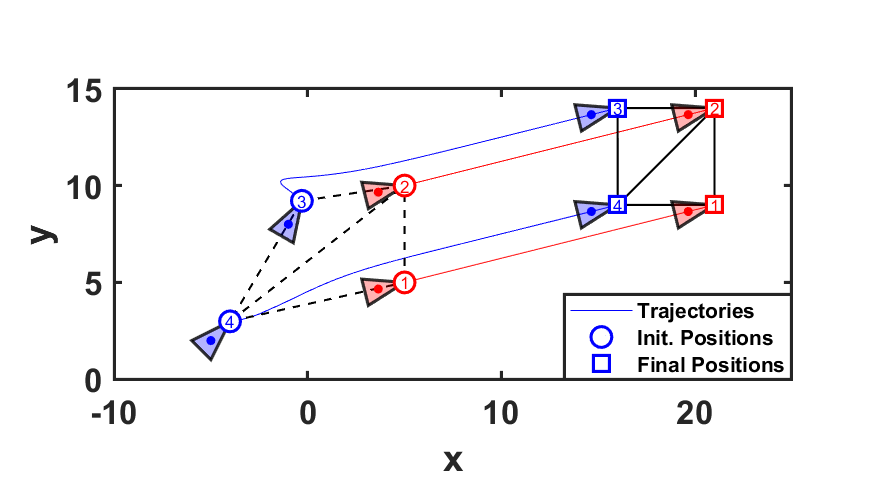}} \\
    \vspace{-4mm}
    \subfloat[]{\includegraphics[width=0.6\linewidth]{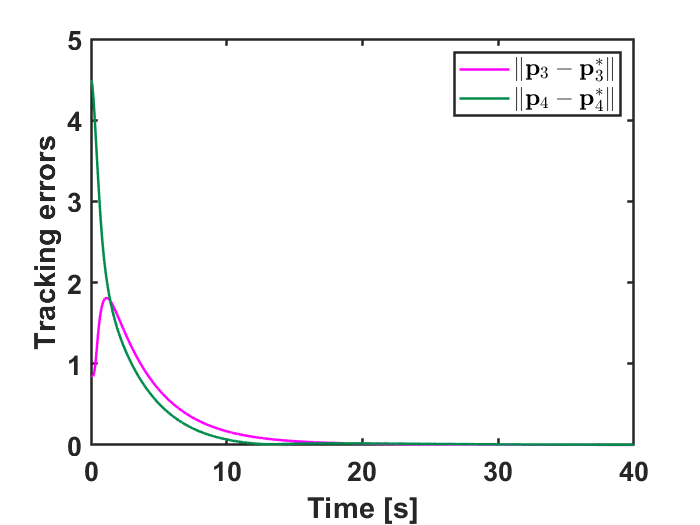}} 
    \caption{Simulation 1: Formation tracking control of four USVs under the control law \eqref{eq:control_velocity}. (a) The trajectories of the followers and leaders after 40 seconds are colored in blue and red, respectively. (b) The tracking errors of each follower.}
    \label{fig:PD}
\end{figure}
\begin{figure}[t!]
    \centering
    \subfloat[]{\includegraphics[width=0.9\linewidth]{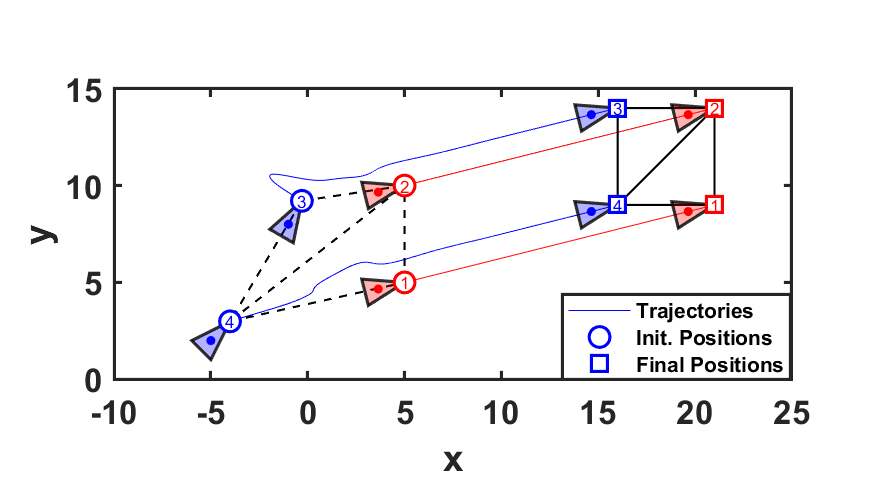}} \\
    \vspace{-4mm}
    \subfloat[]{\includegraphics[width=0.6\linewidth]{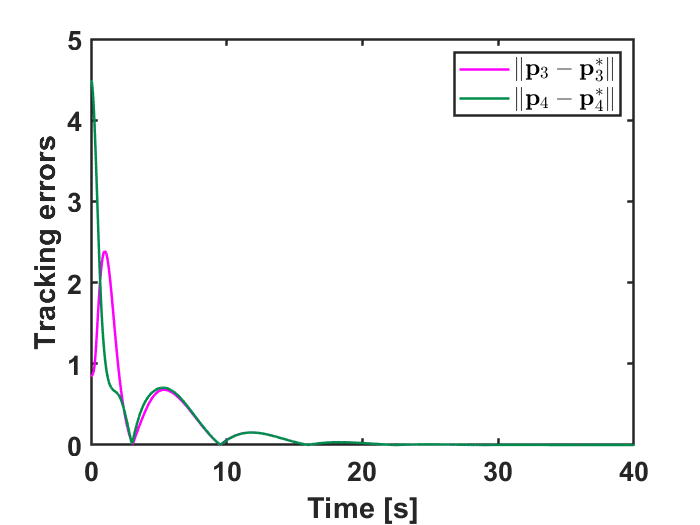}} 
    \caption{Simulation 2: Formation tracking control of four USVs under the control law \eqref{eq:virtual_law}. (a) The trajectories of the followers and leaders after 40 seconds are colored in blue and red, respectively. (b) The tracking errors of each follower.}
    \label{fig:adaptive}
\end{figure}

\section{Conclusion}\label{sec:conclude}
In conclusion, two distributed control laws are proposed for the angle-based formation control problem of underactuated planar agents subject to input disturbances. The desired formation, defined in terms of the agents’ hand points, is assumed to be triangularly angle rigid, while the leaders are assumed to move at constant velocity. The first control law requires relative velocity measurements, whereas the second relies only on relative displacement and local velocity information. The stability of both control laws is rigorously established using different analytical techniques. Potential future directions include extending the results to maneuvering leaders capable of translation, rotation, and scaling, as well as designing control laws with adaptive gains when the disturbance bounds are unknown.





\bibliographystyle{IEEEtran}
\bibliography{ref}

\end{document}